\documentclass[11pt]{article}%
\usepackage{amsfonts}
\usepackage{amsmath}
\usepackage{amssymb}
\usepackage{graphicx}%
\providecommand{\U}[1]{\protect\rule{.1in}{.1in}}
\newtheorem{theorem}{Theorem}
\newtheorem{acknowledgement}[theorem]{Acknowledgement}

\newtheorem{corollary}[theorem]{Corollary}

\newtheorem{definition}[theorem]{Definition}

\newtheorem{lemma}[theorem]{Lemma}

\newtheorem{proposition}[theorem]{Proposition}

\newenvironment{proof}[1][Proof]{\noindent\textbf{#1.} }{\ \rule{0.5em}{0.5em}}
\begin{document}

\title{On Quantum Indeterminacy and the Uncertainty Principle}
\author{Maurice A. de Gosson\thanks{maurice.de.gosson@univie.ac.at}\\University of Vienna\\Faculty of Mathematics, NuHAG\\Vienna \ AUSTRIA,}
\maketitle

\begin{abstract}
We propose a geometric formulation of quantum indeterminacy based on $\hbar
$-polar duality between convex bodies representing position and momentum data.
A pair $(X,P)$ of centrally symmetric convex bodies is said to satisfy the
indeterminacy condition when $X^{\hbar}\subset P$ (equivalently $P^{\hbar
}\subset X$). We show that this condition naturally arises from the geometry
of quantum blobs and is closely related to Hardy's uncertainty principle and
the Donoho--Stark inequalities. Using symplectic capacities and John
ellipsoids, we associate a canonical covariance ellipsoid with every quantum
polar pair and prove that it satisfies the quantum condition. The
Robertson--Schr\"{o}dinger inequalities follow as a consequence. This provides
a non-statistical formulation of quantum indeterminacy from which the usual
uncertainty relations emerge.

\end{abstract}
\date{}

\bigskip




\section{Introduction and Results}

Quantum indeterminacy is an ontological concept: it concerns what exists and
how physical reality is constituted at the quantum level. By contrast, the
uncertainty principle addresses the limits of simultaneous sharp description
of physical quantities (in our case, position and momentum).

Historically, quantum indeterminacy was first formulated by Heisenberg
\cite{Heisenberg} in 1927 in the form $\Delta x\,\Delta p\simeq h$ where
$\Delta x$ and $\,\Delta p$ were viewed as measurement errors; a few years
later Kennard \cite{Kennard} formalized Heisenberg's postulate by proving that
$\Delta x\,\Delta p\geq\hbar/2$ where $\Delta x$ and $\,\Delta p$ were this
time rigorously defined as standard deviations. Robertson and Schr\"{o}dinger
\cite{Robertson} later proved the sharper inequalities%
\begin{equation}
(\Delta x_{j})^{2}(\Delta p_{j})^{2}\geq\Delta(x_{j},p_{j})^{2}+\tfrac{1}%
{4}\hbar^{2}\text{ \ },\text{ \ }1\leq j\leq n \label{RS}%
\end{equation}
which involve both variances and covariances. While these formulations are
mathematically sound and physically successful, they rely on statistical
descriptors whose interpretation as measures of spread is known to be fully
satisfactory essentially only for Gaussian states; this was pointed out and
discussed by Hilgevoord and Uffink \cite{hi02,hiuf85bis,Stanford}. This raises
a natural question: to what extent is the uncertainty principle intrinsically
statistical, and to what extent does it reflect a deeper geometric or
structural property of phase space?

The purpose of the present work is to propose a formulation of quantum
indeterminacy that does not rely on statistical quantities such as means,
variances, or covariances. Instead, we introduce a geometric and topological
framework based on phase-space sets, polar duality, and symplectic invariants.
Our approach is motivated by earlier work \cite{gopolar,FOOP1,point} and by
geometric constraints on phase-space distributions. We consider empirical
position and momentum data, obtained either in separate experimental runs or
through approximate joint measurements. These clouds are ideally represented
by sets $X\subset\mathbb{R}_{x}^{n}$ and $P\subset\mathbb{R}_{p}^{n}$, which
we assume to be convex and centrally symmetric bodies. We propose the
following definition of quantum indeterminacy, which is a mathematical
hypothesis but physically a conjecture which can (in principle) be
experientially tested: \bigskip

\noindent\textbf{Indeterminacy condition.} \textit{Let }$X\subset
\mathbb{R}_{x}^{n}$\textit{ and }$P\subset\mathbb{R}_{p}^{n}$\textit{$_{p}%
^{n}$ be centrally symmetric convex bodies, centered at the origin. We say
that $($}$X,P)$\textit{ satisfies the quantum indeterminacy condition if}
\[
X^{\hbar}\subset P\qquad\mathit{or,equivalently,}\qquad P^{\hbar}\subset X,
\]
\textit{where }$X^{\hbar}$\textit{ is the $\hbar$-polar dual of }%
$X$\textit{,}
\[
X^{\hbar}=\{p\in\mathbb{R}^{n}:\sup_{x\in X}p\cdot x\leq\hbar\}.
\]
\textit{We call }$(X,P)$\textit{ a quantum polar pair.}\medskip

This condition is dimensionally natural, since the pairing $p\cdot x$ appears
in the dimensionless combination $p\cdot x/\hbar$, reflecting the quantum of
action. It provides a non-statistical constraint relating position and
momentum data, expressed purely in geometric terms. It is not, as it might
seem at first glance, pulled out of thin air but motivated by the following
observation (see Section \ref{sec2}): let $\mathrm{QB}(S)=S(B^{2n}(\sqrt
{\hbar}))$ ($S\in\operatorname*{Sp}(n)$, the symplectic group) be a "quantum
blob", \textit{i.e.} a minimum uncertainty cell in phase space. It turns out
that the orthogonal projections $X$ and $P$ of $\mathrm{QB}(S)$ on
$\mathbb{R}_{x}^{n}$ and $\mathbb{R}_{p}^{n}$, respectively, precisely form a
quantum polar pair. As will be seen in Section \ref{sec1} there are other
natural occurrences of quantum polarity, as in Hardy's uncertainty principle
and the Donoho--Stark inequalities.

The assumption that the bodies are centered at the origin is a normalization,
not a physical restriction. If the observed clouds are centered at $x_{0}$ and
$p_{0}$, respectively, one first translates them to $X-x_{0}$ and $P-p_{0}$.
Phase-space translations (and, in particular, momentum boosts) change the
center but not the intrinsic shape of the data or the symplectic uncertainty
constraint. Central symmetry is a genuine hypothesis of the present
formulation; it is precisely the setting in which ordinary polar duality gives
the clean inclusion above. Extensions to non-centrally-symmetric bodies would
require a choice of center or a more general polarity and are not pursued here.

The dimension $n$ should be understood as the number of configuration-space
degrees of freedom. Thus the formalism applies equally to one particle in $n$
spatial dimensions and, after collecting all configuration coordinates, to a
many-particle system with $n$ total degrees of freedom. No tensor-product or
independence assumption is made by the geometric condition itself.

This work is organized as follows:

\begin{itemize}
\item In Section~\ref{sec1}, we review the notion of $\hbar$-polar duality
introduced in our previous work \cite{FOOP1,point,arxiv1,arxiv2}, and
illustrate its usefulness by means of two examples: Hardy's uncertainty
inequalities and the Donoho--Stark concentration principle.

\item In Section~\ref{sec2}, we explain how the notion of quantum blobs,
introduced and studied in \cite{blobs,Bullsci,goluPR}, is related to polar
duality. Quantum blobs are linear symplectic images of phase-space balls of
radius $\sqrt{\hbar}$ and represent minimum-uncertainty phase-space units in
the sense of the Robertson--Schr\"odinger inequalities. We establish the
fundamental correspondence between quantum blobs and generalized coherent
states (Gaussian states).

\item Section~\ref{sec3} is the most technical part of the paper. We show that
the orthogonal projections of a quantum blob onto the $x$- and $p$-spaces form
a quantum polar pair, thereby establishing a direct link between quantum
polarity and the uncertainty principle. This result also reveals a close
analogy between Pauli's reconstruction problem for Gaussian states and the
reconstruction of quantum blobs from their orthogonal projections.

\item In Section~\ref{sec4}, we prove that the John ellipsoid of $X\times P$
contains a quantum blob whenever $P^{\hbar}\subset X$, and is itself a quantum
blob when $P^{\hbar}=X$. This allows us to associate the John ellipsoid of
$X\times P$ with a quantum covariance matrix; the identification is
particularly transparent in the case $P^{\hbar}=X$.

\item In Section~\ref{sec5}, we show that the quantum polarity condition
implies the uncertainty principle in its familiar Robertson--Schr\"odinger
form. We then relate this condition to a symplectic-capacity formulation and
to the usual matrix quantum condition. The notion of symplectic capacity,
which is closely related to Gromov's symplectic non-squeezing theorem, is
briefly reviewed in Appendix~A.

\item Finally, in Section~\ref{secConclusion}, we discuss the scope and
interpretation of our constructions, possible extensions of the present
framework, and technical limitations that remain to be addressed in future work.
\end{itemize}

\paragraph{Notation and terminology}

We denote by $\mathbb{R}_{z}^{2n}\equiv\mathbb{R}_{x}^{n}\times\mathbb{R}%
_{p}^{n}$ the phase space equipped with the standard symplectic form%
\begin{equation}
\sigma=dp\wedge dx=dp_{1}\wedge dx_{1}+\cdot\cdot\cdot+dp_{n}\wedge dx_{n}.
\label{diffnot}%
\end{equation}
In matrix notation $\sigma(z,z^{\prime})=Jz\cdot z^{\prime}=(z^{\prime}%
)^{T}Jz$ where $J=%
\begin{pmatrix}
0_{n\times n} & I_{n\times n}\\
-I_{n\times n} & 0_{n\times n}%
\end{pmatrix}
$ is the standard symplectic matrix.

\section{$\hbar$-Polar Duality: First Steps\label{sec1}}

Let us briefly review the notion of $\hbar$-polar duality we introduced and
related to quantum mechanics in \cite{FOOP1,point,arxiv1,arxiv2}; also see the
recent paper \cite{arxiv3} by Berezovik and Karasev.

\subsection{$\hbar$-polar duality}

Let $X$ be a centrally symmetric convex body in $\mathbb{R}_{x}^{n}$. By
definition the $\hbar$-\emph{polar dual} of $X$ is the subset
\begin{equation}
X^{\hbar}=\{p\in\mathbb{R}_{p}^{n}:px\leq\hbar\text{ \textit{for all} }x\in
X\} \label{omo}%
\end{equation}
of the dual space $\mathbb{R}_{p}^{n}\equiv(\mathbb{R}_{x}^{n})^{\ast}$. In
the mathematical literature one usually chooses $\hbar=1$, in which case one
writes $X^{o}$ for the polar dual; we have $X^{\hbar}=\hbar X^{o}$.

We will from now on abbreviate the terminology by talking about "polarity"
instead of "$\hbar$-polarity" .

The following properties of polarity are obvious:
\begin{align}
\text{\textit{Biduality}}  &  \text{: }(X^{\hbar})^{\hbar}=X~;\\
\text{\textit{Antimonotone}}  &  \text{\textit{: }}X\subset Y\Longrightarrow
Y^{\hbar}\subset X^{\hbar}\\
\text{\textit{Scaling}}  &  \text{: }\det L\neq0\Longrightarrow(LX)^{\hbar
}=(L^{T})^{-1}X^{\hbar}~.
\end{align}
\noindent

\begin{lemma}
Let $B_{X}^{n}(R)$ (\textit{resp}. $B_{P}^{n}(R)$) be the ball $\{x:|x|\leq
R\}$ in $\mathbb{R}_{x}^{n}$ (\textit{resp}. $\{p:|p|\leq R\}$ in
$\mathbb{R}_{p}^{n}$). (i) We have
\begin{equation}
B_{X}^{n}(R)^{\hbar}=B_{P}^{n}(\hbar/R)~. \label{BhR}%
\end{equation}
In particular
\begin{equation}
B_{X}^{n}(\sqrt{\hbar})^{\hbar}=B_{P}^{n}(\sqrt{\hbar})~. \label{bhh}%
\end{equation}
(ii) Let $A=A^{T}$ be positive definite. We have
\begin{equation}
\{x:Ax\cdot x\leq R^{2}\}^{\hbar}=\{p:A^{-1}p\cdot p\leq(\hbar/R)^{2}\}
\label{dualell}%
\end{equation}
and hence
\begin{equation}
\{x:Ax\cdot x\leq\hbar\}^{\hbar}=\{p:A^{-1}p\cdot p\leq\hbar\}~.
\label{dualellh}%
\end{equation}

\end{lemma}

\begin{proof}
Let us show that $B_{X}^{n}(R)^{\hbar}\subset B_{P}^{n}(\hbar/R)$. Let $p\in
B_{X}^{n}(R)^{\hbar}$ and set $x=(R/|p|)p$; we have $|x|=R$ and hence
$px\leq\hbar$, that is $R|p|\leq\hbar$ and $p\in B_{P}^{n}(\hbar/R)$. To prove
the opposite inclusion choose $p\in B_{P}^{n}(\hbar/R)$. We have $|p|\leq
\hbar/R$ and hence, by the Cauchy--Schwarz inequality, $px\leq|x||p|\leq
\hbar|x|/R$, that is $px\leq\hslash$ for all $x$ such that $|x|\leq R$; this
means that $p\in B_{X}^{n}(R)^{\hbar}$. \textit{(ii)} The ellipsoid
$\{x:Ax^{2}\leq R^{2}\}^{\hbar}$ is the image of $B_{X}^{n}(R)$ by the
automorphism $A^{-1/2}$; in view of formula (\ref{dualell}) it follows from
the scaling property and formula (\ref{BhR}) that%
\[
\{x:Ax\cdot x\leq R^{2}\}^{\hbar}=A^{1/2}B_{X}^{n}(R)^{\hslash}=A^{1/2}%
B_{P}^{n}(\hbar/R)
\]
which is equivalent to (\ref{dualell}).
\end{proof}

Defining the Mahler volume \cite{Mahler} of $X$ by
\begin{equation}
\upsilon(X)=\operatorname*{Vol}\nolimits_{n}(X)\operatorname*{Vol}%
\nolimits_{n}(X^{\hbar})=\operatorname*{Vol}\nolimits_{2n}(X\times X^{\hbar})
\label{Mahler}%
\end{equation}
the Blaschke--Santal\'{o} inequality
\[
v(X)\leq(\operatorname*{Vol}\nolimits_{n}(B^{n}(\hbar)))^{2}%
\]
holds; equality holds when $X$ is an ellipsoid. it is conjectured that%
\begin{equation}
\upsilon(X)\geq\frac{(4\hbar)^{n}}{n!} \label{Mahlerbis}%
\end{equation}
(\textquotedblleft Mahler's conjecture\textquotedblright), which can be viewed
as an uncertainty inequality. For a discussion of the related "Mahler
conjecture" see \cite{arxiv3}.

We now introduce the notion of quantum polar pair, which will play a central
role throughout this paper:

\begin{definition}
A pair $(X,P)\ $of centrally symmetric convex bodies $X\subset\mathbb{R}%
_{x}^{n}$ and $P\subset\mathbb{R}_{p}^{n}$ is called a "quantum polar
pair\textquotedblright\ if we have
\begin{equation}
X^{\hbar}\subset P\text{ or,equivalently, }P^{\hbar}\subset X. \label{XP}%
\end{equation}
When equality occurs, that is if $X^{\hbar}=P$, we say that the quantum polar
pair $(X,P)$ is saturated.
\end{definition}

Notice that, due to the antimonotonicity of polar duality, if $(X,P)$ be a
quantum pair then so is $(Y,Q)$ if $Y,Q$ be symmetric convex \ bodies such
that $X\subset Y$ and $P\subset Q$.

Here is the central result relating our previous constructions to the notion
of symplectic capacity. It relies on a very deep and difficult proof in
\cite{Artstein} in an attempt to prove the Mahler conjecture. (See Appendix A
for a review of the notion of symplectic capacity, and related concepts).

\begin{theorem}
\label{Thm5}Let $(X\times X^{\hbar})$ be a saturated quantum polar pair. We
have%
\begin{equation}
c_{^{\text{HZ}}}(X\times X^{\hbar})=4\hbar\label{clh}%
\end{equation}
where $c_{^{\text{HZ}}}$ is the Hofer--Zehnder capacity. For a general quantum
dual pair \ $(X\times P)$, we have
\begin{equation}
c_{^{\text{HZ}}}(X\times P)=4\lambda_{\max}\hbar\label{yaron1}%
\end{equation}
where $\lambda_{\max}\geq1$ is the number
\begin{equation}
\lambda_{\max}=\max\{\lambda>0:\lambda X^{\hbar}\subset P\}. \label{amax}%
\end{equation}

\end{theorem}

\begin{proof}
For a detailed proof when $X$ is an ellipsoid see \cite{gopolar}; also Prop. 3
in \cite{ACHAPOLAR}). In the general case one has to use results in
\cite{Artstein} Artstein-Avidan \textit{et al}. show that for $\hbar=1$ we
have $c_{\max}(X\times X^{1})=4$. Using the obvious relation $X=\hbar X^{1}$
we have, by the conformality property of symplectic capacities,%
\begin{align*}
c_{^{\text{HZ}}}(X\times X^{\hbar})  &  =c_{^{\text{HZ}}}(\hbar^{1/2}%
(\hbar^{-1/2}X\times\hbar^{1/2}X^{1}))\\
&  =\hbar c_{^{\text{HZ}}}(\hbar^{-1/2}X\times\hbar^{1/2}X^{1})\\
&  =\hbar c_{^{\text{HZ}}}(X\times X^{\hbar})
\end{align*}
the last equality because
\[
\hbar^{-1/2}X\times\hbar^{1/2}X^{1}=M_{\hbar^{1/2}}(X,X^{1})
\]
with $M_{\hbar^{1/2}}\in\operatorname*{Sp}(n)$; hence $cc_{^{\text{HZ}}%
}(X\times X^{\hbar})=$ $4\hbar$. Formula (\ref{clh}) follows since using the
symplectic invariance of symplectic capacities. A similar argument using in
\cite{Artstein} leads to (\ref{yaron1}).
\end{proof}

\begin{corollary}
Each saturated quantum dual pair $(X,X^{\hbar})$ is such that the infimum of
the areas \textit{of the smallest disk containing} $\pi_{j}(X\times X^{\hbar
})\}$is $\geq4\hbar$ .
\end{corollary}

\begin{proof}
Definition (\ref{cmax}) of $c_{\max}(\Omega)$ is equivalent to the following
statement: let $\pi_{j}:\mathbb{R}^{2n}\rightarrow\mathbb{R}^{2}$ be the
orthogonal projection $\pi(x_{1},y_{1},\dots,x_{n},y_{n})=(x_{j},p_{j})$. We
have
\[
c_{\max}(\Omega)=\inf_{f\in\mathrm{Symp}(n)}\{\text{\textit{area of the
smallest disk containing} }\pi_{j}(f\Omega))\}.
\]
for any $j=1,...,n$. The result follows from fact that $c_{\max}\geq
c_{^{\text{HZ}}}$. The corollary follows, taking $\Omega=X\times X^{\hbar}$.
\end{proof}

\subsection{Polar duality and function concentration}

The following justifies the use of polar duality to measure indeterminacy in a
geometric say. Assume that $\mathcal{C}_{X}$ and $\mathcal{C}_{P}$ measurable
sets. We will say that a function $\psi\in L^{2}(\mathbb{R}^{n})$ such tat
$||\psi||_{L^{2}}=1$ and its Fourier transform $\widehat{\psi}$ are,
respectively, $\varepsilon_{X}$ and $\varepsilon_{P}$ concentrated on
$\mathcal{C}_{X}$ and $\mathcal{C}_{P}$ if we have
\begin{equation}
\int_{\mathcal{C}_{X}}|\psi(x)|^{2}dx\geq1-\varepsilon_{x}^{2}\text{ }%
\int_{\mathcal{C}P}|\widehat{\psi}(p)|^{2}dp\geq1-\varepsilon_{P}^{2}.
\label{ksiksi}%
\end{equation}
The Donoho--Stark uncertainty principle says that if for $\varepsilon
_{X}+\varepsilon_{P}<1$ then
\begin{equation}
\operatorname*{Vol}\nolimits_{2n}(\mathcal{C}_{X}\times\mathcal{C}_{P}%
)\geq(2\pi\hbar)^{n}(1-\varepsilon_{X}-\varepsilon_{P})^{2} \label{DS}%
\end{equation}
(see \cite{Gro} for an elegant proof).

\begin{proposition}
Assume that $X\subset\mathcal{C}_{X}$. If
\begin{equation}
\int_{X}|\psi(x)|^{2}dx\geq1-\varepsilon_{x}^{2}\text{ \ and \ }\int%
_{X^{\hbar}}|\widehat{\psi}(p)|^{2}dp\geq1-\varepsilon_{P}^{2} \label{blasa}%
\end{equation}
then we must have%
\begin{equation}
0\leq(2\pi\hbar)^{n}(1-\varepsilon_{X}-\varepsilon_{P})^{2}\leq\frac{(\pi
\hbar)^{n}}{\Gamma(n/2+1)^{2}}. \label{dostpolar}%
\end{equation}

\end{proposition}

\begin{proof}
The Mahler volume $\upsilon(X)=\operatorname*{Vol}\nolimits_{n}%
(X)\operatorname*{Vol}\nolimits_{n}(X^{\hbar})$ satisfies the
Blaschke--Santal\'{o} inequity%
\[
\upsilon(X)\leq\left[  \operatorname*{Vol}\nolimits_{n}B^{n}(\sqrt{\hbar
})\right]  ^{2}=\frac{(\pi\hbar)^{n\Gamma}}{\Gamma(n/2+1)^{2}}.
\]
It follows from the Donoho--Stark principle that%
\[
\left[  \operatorname*{Vol}\nolimits_{n}B^{n}(\sqrt{\hbar})\right]  ^{2}%
.\geq(2\pi\hbar)^{n}(1-\varepsilon_{X}-\varepsilon_{P})^{2}%
\]
hence the inequality (\ref{dostpolar}).
\end{proof}

It follows from For Stirling's formula,that for $n\rightarrow\infty$ we have
\begin{equation}
0<(2\pi\hbar)^{n}(1-\varepsilon_{X}-\varepsilon_{P})^{2}\lesssim\frac{1}{\pi
n}\left(  \frac{2e\pi\hbar}{n}\right)  ^{n}. \label{asy,pt}%
\end{equation}
It follows that for large $n$ or small$\hbar$ we have $\varepsilon
_{X}+\varepsilon_{P}\cong1$, which illustrates the trade-off between the
concentration of a function and its Fourier transform.

Here is another justification of the use of polar duality in the formulation
of uncertainty principles\cite{ACHAPOLAR,golulett}. In \cite{ACHAPOLAR,goluPR}
we proved:

\begin{proposition}
\label{prophardy}Let $A$ and $B$ be two real positive definite matrices and
$\psi\in L^{2}(\mathbb{R}^{n})$, $\psi\neq0$. Assume that there exist
constants $C_{A},C_{B}>0$ such that
\begin{equation}
|\psi(x)|\leq C_{A}e^{-\tfrac{1}{2\hbar}Ax\cdot x}\ \text{and}\ |\widehat{\psi
}(p)|\leq C_{B}e^{-\tfrac{1}{2\hbar}Bp\cdot p}. \label{Hardy}%
\end{equation}
Then: (i) The eigenvalues $\lambda_{j}$, $j=1,...,n$, of $AB$ are $\leq1$;
(ii) If $\lambda_{j}=1$ for all $j$, then $\psi(x)=Ce^{-\frac{1}{2\hbar}%
Ax^{2}}$ for some some complex constant $C$; (iii) If $\lambda_{j}<1$ for some
index $j$ then there exist functions $\psi(x)=Q(x)e^{-\tfrac{1}{2\hbar}Ax^{2}%
}$ ($Q$ a polynomial) satisfying the (\ref{Hardy}).
\end{proposition}

This multidimensional version of the Hardy uncertainty principle can be
elegantly expressed in terms of polar duality:

\begin{corollary}
\label{corhardy}Assume that \ $\psi\in L^{2}(\mathbb{R}^{n})$ satisfies the
Hardy inequalities (\ref{Hardy}). Let us set
\[
X_{A}=\{x:Ax\cdot x\leq\hbar\}\text{ \ },\text{ \ }P_{B}=\{p:Bp\cdot
p\leq\hbar\}.
\]
$~$Suppose (i) \ The Hardy inequalities. are satisfied if and only if
$(X_{A},P_{B})$ is a quantum dual pair, that is $X_{A}^{\hbar}\subset P_{B}$.
(ii) When $X_{A}^{\hbar}=P_{B}$ then $\psi(x)=Ce^{-\frac{1}{2\hbar}Ax^{2}}$
for some complex number $C\neq0$.
\end{corollary}

\begin{proof}
(i) The condition $X_{A}^{\hbar}\subset P_{B}$ is equivalent to $A^{-1}\geq B$
(L\"{o}wner ordering), this is the same as saying that the eigenvalues of $AB$
are $\leq1$ and the claim the follows from (i) in the proposition above. (ii)
In this case all the eigenvalues of $AB$ are equal to one hence (ii) in the
proof of Proposition \ref{prophardy} applies.
\end{proof}

\section{Quantum Blobs and Polar Duality\label{sec2}}

The symplectic group $\operatorname*{Sp}(n)$ is the group of all linear
automorphisms of $\mathbb{R}^{2n}=T^{\ast}\mathbb{R}^{n}$ preserving the
symplectic form $\sigma$. in coordinates: $S\in\operatorname*{Sp}(n)$ if and
only $S^{T}JS=SJS^{T}=J$. A block-matrix
\[
S=%
\begin{pmatrix}
A & B\\
C & D
\end{pmatrix}
\]
is symplectic if and only if
\begin{align}
A^{T}C\text{, }B^{T}D\text{ \ symmetric, and }A^{T}D-C^{T}B  &
=I\label{cond1}\\
AB^{T}\text{, }CD^{T}\text{ \ symmetric, and }AD^{T}-BC^{T}  &  =I\text{.}
\label{cond2}%
\end{align}
It follows from the second of these sets of conditions that the inverse of $S$
is then
\begin{equation}
S^{-1}=%
\begin{pmatrix}
D^{T} & -B^{T}\\
-C^{T} & A^{T}%
\end{pmatrix}
\text{.} \label{cond3}%
\end{equation}
The symplectic group is generate by $J$ and the automorphisms
\[
M_{L}=%
\begin{pmatrix}
L^{-1} & 0\\
0 & L^{T}%
\end{pmatrix}
\text{ \ , \ }V_{-P}=%
\begin{pmatrix}
I & 0\\
P & I
\end{pmatrix}
\]
where $L\in GL(n,\mathbb{R})$) and $P\in\operatorname*{Sym}(n,\mathbb{R})$.

\subsection{Quantum blobs}

We introduced the notion of quantum blob in \cite{blobs} as a minimum
uncertainty phase space cell compatible with the action of $\operatorname*{Sp}%
(n)$. The concept was further developed in, for instance \cite{goluPR}. \ By
definition, A (centered) quantum blob is a symplectic ball with radius
$\sqrt{\hbar}$ that is an ellipsoid of $\mathbb{R}^{2n}$ of the type
\begin{equation}
\operatorname*{QB}(S)=S(B^{2n}(\sqrt{\hbar})) \label{blob}%
\end{equation}
where $S\in\operatorname*{Sp}(n)$. Quantum blobs with arbitrary center $z_{0}$
are defined by translating $\operatorname*{QB}(S)$. we will only consider the
centered case.

The volume of a quantum blob quickly decreases with the dimension of phase
space: since symplectic automorphisms are volume-preserving we have%
\begin{equation}
\operatorname*{Vol}(\operatorname*{QB}(S))=\operatorname*{Vol}(B^{2n}%
(\sqrt{\hbar}))=\frac{1}{n!}\left(  \frac{h}{2}\right)  ^{n}. \label{volblob}%
\end{equation}

It turns out that quantum blobs can be obtained using just two simple types of
symplectic transformations. Recall the pre-Iwasawa decomposition of a
symplectic matrix \cite{iwa,Birk,Houde}%
\[
S=%
\begin{pmatrix}
A & B\\
C & D
\end{pmatrix}
.
\]
There exist unique matrices $P,L\in\operatorname*{Sym}(n,\mathbb{R})$, $L=>0$,
and $R\in\operatorname*{Sp}(n)\cap O(2n)$ such that
\begin{equation}
S=V_{-P}M_{L}R. \label{iwa1}%
\end{equation}
These matrices are explicitly given by the formulas%
\begin{gather}
L=(AA^{T}+BB^{T})^{-1/2}\label{pl1}\\
P=(CA^{T}+DB^{T})(AA^{T}+BB^{T})^{-1}. \label{pl2}%
\end{gather}
The matrix $R$ is a symplectic rotation: writing
\[
R=%
\begin{pmatrix}
E & F\\
-F & E
\end{pmatrix}
\in\operatorname*{Sp}(n)\cap O(2n,\mathbb{R})
\]
the $n\times n$ blocks $E$ and $F$ are given by%
\begin{equation}
E=(AA^{T}+BB^{T})^{-1/2}A\text{ \ },\text{ \ }F=(AA^{T}+BB^{T})^{-1/2}B.
\label{unixy}%
\end{equation}
It follows from the pre-Iwasawa factorization result that:

\begin{proposition}
\label{lemqb}Every quantum blob $\operatorname*{QB}(S)$ \ can be written
uniquely as
\begin{equation}
\operatorname*{QB}(S)=V_{-P}M_{L}(B^{2n}(\sqrt{\hbar})) \label{blobml}%
\end{equation}
where $L$ and $P$ are defined using formulas (\ref{pl1}) and (\ref{pl2}) above.
\end{proposition}

\begin{proof}
It is as straightforward consequence of the factorization (\ref{iwa1}) since
$B^{2n}(\sqrt{\hbar})$ is invariant under rotations. The uniqueness follows
from the observation that if $V_{-P}M_{L}=V_{-P^{\prime}}^{\prime}%
M_{L^{\prime}}$ then $V_{P^{\prime}-P}=M_{(L^{\prime})^{-1}L}$ and hence
$P=P^{\prime}$ and $L=L^{\prime}$.
\end{proof}

\subsection{Identifying quantum blobs and generalized Gaussians}

Let us briefly review the main properties of the Wigner function (or
transform) \cite{WIGNER,Folland}. By definition, the Wigner function of
$\psi\in L^{2}(\mathbb{R}^{n})$ is the real function $W\psi\in L^{2}%
(\mathbb{R}^{2n})$ \ defined by
\begin{equation}
W\psi(z)=\left(  \tfrac{1}{2\pi\hbar}\right)  ^{n}\int_{\mathbb{R}^{n}%
}e^{-\frac{i}{\hbar}py}\psi(x+\tfrac{1}{2}y)\overline{\psi(x-\tfrac{1}{2}%
y)}dy~. \label{wigner}%
\end{equation}
\ \ \ In the case
\[
\psi\in L^{1}(\mathbb{R}^{n})\cap L^{2}(\mathbb{R}^{n})\text{ \ },\text{
\ }||\psi||_{L^{2}}=1
\]
the function $W\psi$ is a quasi-probability function, in the sense that
\[
\int_{\mathbb{R}^{n}}W\psi(x,p)dp=|\psi(x)|^{2}\text{ \ },\text{ \ }%
\int_{\mathbb{R}^{n}}W\psi(x,p)dx=|\widehat{\psi}(x)|^{2}%
\]
where $\widehat{\psi}$ is the Fourier transform
\[
\widehat{\psi}(p)=\left(  \tfrac{1}{2\pi\hbar}\right)  ^{n/2}\int%
_{\mathbb{R}^{n}}e^{-\frac{i}{\hbar}px}\psi(x)dx.
\]
In particular,%
\[
\int_{\mathbb{R}^{2n}}W\psi(x,p)dpdx=1.
\]
We also recall \cite{Birk,WIGNER} the symplectic covariance formula
\begin{equation}
W(\widehat{S}^{-1}\psi)=W\psi\circ S \label{sycowig}%
\end{equation}
for all metaplectic operators $\widehat{S}$ with projection $S\in
\operatorname*{Sp}(n).$

We consider now generalized centered Gaussian functions
\begin{equation}
\psi_{W,Y}(x)=\left(  \tfrac{\det W}{(\pi\hbar)^{n}}\right)  ^{1/4}%
e^{-\frac{1}{2\hbar}(W+iY)x\cdot x} \label{squeezed}%
\end{equation}
where $W,Y\in\operatorname*{Sym}(n)$ ,and $>0$. These Gaussians can be
obtained from the standard Gaussian
\begin{equation}
\phi_{0}(x)=\psi_{I,0}(x)=(\pi\hbar)^{-n/4}e^{-|x|^{2}/2\hbar}
\label{standard}%
\end{equation}
using the elementary metaplectic operators \cite{Birk,Birkbis}%
\begin{gather}
\widehat{J}\psi(x)=\left(  \tfrac{1}{2\pi i\hbar}\right)  ^{n/2}%
\int_{\mathbb{R}^{n}}e^{-\frac{i}{\hbar}x\cdot x^{\prime}}\psi(x^{\prime
})dx^{\prime}.\label{mp1}\\
\widehat{M}_{L,m}\psi(x)=i^{m}\sqrt{|\det L|}\psi(Lx)\text{\ },\text{ \ }%
m\pi=\arg\det L,\label{mp2}\\
\widehat{V}_{P}\psi(x)=e^{-\frac{i}{2\hbar}Px\cdot x}\psi(x)\text{ }
\label{mp3}%
\end{gather}
covering $J,M_{L},V_{P}$, respectively. This follows from the elementary
observation that
\begin{equation}
\psi_{W,Y}=\widehat{V}_{Y}\widehat{M}_{W^{1/2}}\phi_{0} \label{,pfio}%
\end{equation}
where $\phi_{0}$ is the standard Gaussian. Note that the Wigner function of
$\psi_{W,Y}=V_{Y}M_{W^{-1/2}}\phi_{0}$ is given by the well-known
\cite{Birk,Birkbis,Folland} formula, which immediately follows from our
constructions above:%
\begin{equation}
W\psi_{W,Y}(z)=(\pi\hbar)^{-n}e^{-\frac{1}{\hbar}Gz\cdot z} \label{wigpsi}%
\end{equation}
where $G$ is the positive definite symplectic matrix
\begin{equation}
G=%
\begin{pmatrix}
W+YW^{-1}Y & YW^{-1}\\
W^{-1}Y & W^{-1}%
\end{pmatrix}
. \label{G}%
\end{equation}

\begin{theorem}
\label{ThmBlob}Let $\operatorname*{Blob}\nolimits_{0}(n)$ be the set of all
quantum blobs $\operatorname*{QB}(S)$ and $\operatorname*{Gauss}%
\nolimits_{0}(n)$ the set of all generalized Gaussians $\psi_{W,Y}$. The
mapping
\begin{equation}
\Gamma:\operatorname*{Blob}\nolimits_{0}(n)\longrightarrow
\operatorname*{Gauss}\nolimits_{0}(n) \label{ggablob}%
\end{equation}
defined by
\begin{equation}
\Gamma:V_{Y}M_{W^{-1/2}}B^{2n}(\sqrt{\hbar}))\longmapsto\widehat{V}%
_{Y}\widehat{M}_{W^{-1/2}W,Y}\phi_{0} \label{gablobis}%
\end{equation}
is a bijection. Explicitly:%
\begin{equation}
\Gamma(\operatorname*{QB}(V_{Y}M_{W^{-1/2}}))(z)=(\pi\hbar)^{-n/4}\det
W^{-1/2}e^{-\frac{1}{2\hbar}(W+iY)x\cdot x.}. \label{explicit}%
\end{equation}

\end{theorem}

\begin{proof}
Set $S_{W,Y}=V_{Y}M_{W^{-1/2}}$. Let us thus show that
\begin{gather*}
\Gamma:\operatorname*{Blob}\nolimits_{0}(n)\longrightarrow
\operatorname*{Gauss}\nolimits_{0}(n)\\
S_{W,Y}B^{2n}(\sqrt{\hbar})\longmapsto\widehat{S}_{W,Y}\phi_{0}%
\end{gather*}
is a bijection. Firstly, $\Gamma$ is a well-defined mapping since every
quantum blob can be written using the pre-Iwasawa factorization as
(\ref{iwa1}) as%
\[
\operatorname*{QB}(S_{W,Y})=S_{W,Y}(B^{2n}(\sqrt{\hbar})=V_{Y}M_{W^{-1/2}%
}(B^{2n}(\sqrt{\hbar})).
\]
Similarly every $\psi_{W,Y}$ can be written as $\psi_{W,Y}=\widehat{S}%
_{W,Y}\phi_{0}$, showing at the same time that $\Gamma$ is surjective. To show
that $\Gamma$ is \ bijection there remains to \ prove injectivity, that is if
$\widehat{S}_{W,Y}\phi_{0}=\widehat{S}_{W^{\prime},Y^{\prime}}^{\prime}%
\phi_{0}$ then
\[
S_{W,Y}B^{2n}(\sqrt{\hbar}))=S_{W^{\prime},Y^{\prime}}B^{2n}(\sqrt{\hbar})).
\]
In view of the rotational symmetry of the standard Gaussian $\phi_{0}$ we must
have $\widehat{S}_{W,Y}=\widehat{S}_{W^{\prime},Y^{\prime}}\widehat{R}$ where
$\widehat{R}\in\operatorname*{Mp}(n)$ \ covers a symplectic rotation
$R\in\operatorname*{Sp}(n)\cap O(2n,\mathbb{R})$, hence $S_{W^{\prime
},Y^{\prime}}=S_{W,Y}R$ and the injectivity follows since $R(B^{2n}%
(\sqrt{\hbar})))=B^{2n}(\sqrt{\hbar}))$.
\end{proof}

\section{Projecting Quantum Blobs\label{sec3}}

This section connects quantum polar parity to quantum blobs: in addition to
being minimal uncertainty phase space cells, the orthogonal projections on
$\mathbb{R}_{x}^{n}$ and $\mathbb{R}_{p}^{n}$ are polar dual pairs.

\subsection{The projection theorem}

The following result is central to this work; it will show that quantum
polarity is directly related to the uncertainty principle:

\begin{theorem}
\label{Thm1}(i) The orthogonal projections $X$ and $P$ of a quantum blob
$\operatorname*{QB}(S)$ on the $x$ and $p$ spaces form a quantum polar pair;
i.e:%
\begin{equation}
\left(  \Pi_{X}\operatorname*{QB}(S)\right)  ^{\hbar}\subset\Pi_{P}%
\operatorname*{QB}(S). \label{pixp1}%
\end{equation}
(ii) The equality
\begin{equation}
\left(  \Pi_{X}\operatorname*{QB}(S)\right)  ^{\hbar}=\Pi_{P}%
\operatorname*{QB}(S) \label{ppicpus}%
\end{equation}
occurs if and only $S=M_{L}$ for some $L\in GL(n,\mathbb{R})$.
\end{theorem}

\begin{proof}
The ellipsoid $\operatorname*{QB}(S)=S(B^{2n}(\sqrt{\hbar}))$ consists of all
$z\in\mathbb{R}^{2n}$ such that $Mz\cdot z\leq\hbar$ where
\[
M=(SS^{T})^{-1}\in\operatorname*{Sp}(n).
\]
The matrix $M$ is symmetric and positive definite; we write it as above in
block-matrix form
\[
M=%
\begin{pmatrix}
M_{XX} & M_{XP}\\
M_{XP}^{T} & M_{PP}%
\end{pmatrix}
\]
\ with $M_{XX}>0$, $M_{PP}>0$. Since $M\in\operatorname*{Sp}(n)$ and $M=M^{T}$
we have $MJM=J$ which is equivalent to the symplectic relations
\begin{gather}
M_{XX}M_{PP}-M_{XP}^{2}=I_{n\times n}\label{mc1}\\
M_{XX}M_{PX}=M_{XP}M_{XX}\label{mc2}\\
M_{PX}M_{PP}=M_{PP}M_{XP}~. \label{mc3}%
\end{gather}
The projections $X=\Pi_{X}\operatorname*{QB}(S)$ and $P=\Pi_{P}%
\operatorname*{QB}(S)$ of $S(B^{2n}(\sqrt{\hbar}))$ are given by
\begin{align}
X  &  =\{x:(M/M_{PP})x\cdot x\leq\hbar\}\label{schurco1}\\
\text{ \ }P  &  =\{p:(M/M_{XX})p\cdot p\leq\hbar\} \label{schurco2}%
\end{align}
where $M/M_{PP}$ and $M/M_{XX}$ are the Schur complements%
\begin{align}
M/M_{PP}  &  =M_{XX}-M_{XP}M_{PP}^{-1}M_{PX}\label{schur1}\\
M/M_{XX}  &  =M_{PP}-M_{PX}M_{XX}^{-1}M_{XP} \label{schur2}%
\end{align}
are the Schur complements; we have $M/M_{PP}>0$ and $M/M_{XX}>0$. It follows
that we have
\[
X^{\hbar}=\{p:(M/M_{PP})^{-1}p\cdot p\leq\hbar\}
\]
hence the condition $X^{\hbar}\subset P$ is equivalent to%
\begin{equation}
(M/M_{PP})^{-1}\geq M/M_{XX}~. \label{RPPXX}%
\end{equation}
Since $M$ is symplectic and symmetric we have%
\[
M=^{-1}-JMJ=%
\begin{pmatrix}
M_{PP} & -M_{XP}^{T}\\
-M_{XP} & M_{XX}%
\end{pmatrix}
\]
and also, using the block inversion formula,
\[
M^{-1}=%
\begin{pmatrix}
\left(  M_{XX}-M_{XP}M_{PP}^{-1}M_{PX}\right)  ^{-1} & \ast\\
\ast & \left(  M_{PP}-M_{PX}M_{XX}^{-1}M_{XP}\right)  ^{-1}%
\end{pmatrix}
\]
and hence, comparing both formulas,%
\begin{align}
M/M_{PP}  &  =M_{XX}-M_{XP}M_{PP}^{-1}M_{PX}=M_{PP}^{-1}\label{MX}\\
M/M_{XX}  &  =M_{PP}-M_{PX}M_{XX}^{-1}M_{XP}=M_{XX}^{-1} \label{MP}%
\end{align}
and therefore
\begin{equation}
X^{\hbar}=\{p:M_{PP}p\cdot p\leq\hbar\}\text{ \ },\text{\ \ }P=\{p:(M_{XX}%
^{-1}p\cdot p\leq\hbar\} \label{MPX}%
\end{equation}
so the condition $X^{\hbar}\subset P$ is equivalent to $M_{PP}^{-1}\leq
M_{XX}$, and the latter holds follows from the inequality
\[
M_{PP}^{-1}=M_{XX}-M_{XP}M_{PP}^{-1}M_{PX}\leq M_{XX}.
\]
Hence equality $M_{PP}^{-1}=M_{XX}$ holds if and only if \ $M_{XP}=0$ in which
case $X^{\hbar}=P.$
\end{proof}

The orthogonal projections of a quantum blob do not define it uniquely. This
is closely related to Pauli's famous reconstruction problem, which we briefly
discuss below.

\subsection{Analogy with Pauli's reconstruction problem}

Pauli's reconstruction problem \cite{Pauli} arises from the following
question: "Does the knowledge of $|\psi(x)|$ and $|\widehat{\psi}(p)|$ of the
Gaussian
\[
\psi_{W,Y}(x)=\left(  \tfrac{\det W}{(\pi\hbar)^{n}}\right)  ^{1/4}%
e^{-\frac{1}{2\hbar}(W+iY)x\cdot x}%
\]
allow one to reconstruct this function?" As is well-known, the answer is "no".
Consider the case of a generalized Gaussian $\psi_{W,Y}$; it has covariance
matrix%
\begin{equation}%
\begin{pmatrix}
\Sigma_{XX} & \Sigma_{XP}\\
\Sigma_{XP}^{T} & \Sigma_{PP}%
\end{pmatrix}
=\dfrac{\hbar}{2}G^{-1} \label{covG}%
\end{equation}
where $G=%
\begin{pmatrix}
G_{XX} & G_{XP}\\
G_{XP}^{T} & G_{PP}%
\end{pmatrix}
$ i the symplectic matrix (see \ref{G}), whose inverse is
\[
G^{-1}=%
\begin{pmatrix}
W & -W^{-1}Y\\
-YW^{-1} & ^{-1}W+YW^{-1}Y
\end{pmatrix}
.
\]
Since $G^{-1}$ is symplectic we must have%
\begin{equation}
\Sigma_{XX}\Sigma_{PP}-(\Sigma_{XP})^{2}=\frac{\hbar^{2}}{4}I_{n\times n}
\label{condsig}%
\end{equation}
(which is a matrix form of the Arborists inequalities). Now, a straightforward
computation shows that we have%
\begin{align}
|\psi_{W,Y}(x)|^{2}  &  =\left(  \tfrac{1}{2\pi}\right)  ^{n/2}(\det
\Sigma_{XX})^{-1/2}\exp\left(  -\frac{1}{2}\Sigma_{XX}^{-1}x\cdot x\right)
\label{margauss1}\\
|\widehat{\psi}_{W,Y}(p)|^{2}  &  =\left(  \tfrac{1}{2\pi}\right)  ^{n/2}%
(\det\Sigma_{PP})^{-1/2}\exp\left(  -\frac{1}{2}\Sigma_{PP}^{-1}p\cdot
p\right)  \label{margauss2}%
\end{align}
and these relations do not determine uniquely $\Sigma_{XP}$; in fact
(\ref{condsig}) leads to many solutions and hence to many functions
$\psi_{W,Y}$ (called by Corbett \cite{Corbett} "Pauli partners").

In fact, $_{W,Y}($ is associated, by the bijection $\Gamma$ described in
Theorem \ref{ThmBlob} to the quantum blob $\{z:Gz\cdot z\leq\hbar\}$; in view
of formulas (\ref{schurco1}), (\ref{schurco2}), and (\ref{MPX}) the orthogonal
projection of this quantum blob are
\begin{align}
X  &  =\{x:G_{PP}x\cdot x\leq\hbar\}=\{p:W^{-1}x\cdot x\leq\hbar\}\text{ \ }\\
\text{\ \ }P  &  =\{p:(G_{XX}^{-1}p\cdot p\leq\hbar\})
\end{align}
leading to the equalities
\[
\Sigma_{XX}=W^{-1}\text{ \ , \ }\Sigma_{PP}=W+YW^{-1}Y.
\]

\section{John Ellipsoids and Quantum Polar Duality\label{sec4}}

The John ellipsoid $\operatorname{John}(K)$ of a convex set $K$ in an
Euclidean space is the unique ellipsoid with largest volume contained in that
set \cite{Ball}. It is equivariant under both translations:
$\operatorname{John}(K+z)=\operatorname{John}(K)+z$ and linear automorphisms
$A$: $A\operatorname{John}(K)=\operatorname{John}(AK)$.

\subsection{The saturated case of $X\times X^{\hbar}$}

Here we study the problem of which covariance matrices can be reconstructed
from quantum polar pairs $(X,P)$,$X^{\hbar}\subset P$. \ We write, , as usual,%
\begin{align*}
X  &  =\{p:Ax\cdot x\leq\hbar\}=A^{-1/2}(B_{P}^{n}(\sqrt{\hbar})\\
X^{\hbar}  &  =\{p:A^{-1}p\cdot p\leq\hbar\}=A^{1/2}(B_{P}^{n}(\sqrt{\hbar
}))\\
P  &  =\{p:Bp\cdot p\leq\hbar\}=B^{-1/2}(B_{P}^{n}(\sqrt{\hbar})).
\end{align*}
We begin with the "saturated" case, which is rather straightforward. For this
purpose we will need the following straightforward result:

\begin{lemma}
\label{LemLoewner}The John ellipsoid of the product $B_{X}^{n}(\sqrt{\hbar
})\times B_{P}^{n}(\sqrt{\hbar})$ is the phase space ball $B^{2n}(\sqrt{\hbar
})$:%
\[
\operatorname{John}(B_{X}^{n}(\sqrt{\hbar})\times B_{P}^{n}(\sqrt{\hbar
}))=B^{2n}(\sqrt{\hbar}).
\]

\end{lemma}

\begin{proof}
The inclusion
\begin{equation}
B^{2n}(\sqrt{\hbar})\subset B_{X}^{n}(\sqrt{\hbar})\times B_{P}^{n}%
(\sqrt{\hbar}) \label{incl}%
\end{equation}
is obvious, and we cannot have
\[
B^{2n}(R)\subset B_{X}^{n}(\sqrt{\hbar})\times B_{P}^{n}(\sqrt{\hbar})
\quad\text{if } R>\sqrt{\hbar}.
\]
Assume now that the John ellipsoid $\Omega_{\mathrm{John}}$ of
\[
\Omega=B_{X}^{n}(\sqrt{\hbar})\times B_{P}^{n}(\sqrt{\hbar}))
\]
is defined by
\[
Ax\cdot x+Bx\cdot p+Cp\cdot p\leq\hbar
\]
where $A,C>0$ and $B$ are real $n\times n$ matrices. Since $\Omega$ is
invariant by the transformation $(x,p)\longmapsto(p,x)$ so is $\Omega
_{\mathrm{John}}$ and we must thus have $A=C$ and $B=B^{T}$. Similarly,
$\Omega$ being invariant by the partial reflection $(x,p)\longmapsto(-x,p)$ we
get $B=0$ so $\Omega_{\mathrm{John}}$ is defined by $Ax\cdot x+Ap\cdot
p\leq\hbar$. The next step is to observe that $\Omega$ and hence
$\Omega_{\mathrm{John}}$ are invariant under all transformations
$(x,p)\longmapsto(Hx,Hp)$ where $H\in O(n,\mathbb{R})$, so we must have
$AH=HA$ for all $H\in O(n,\mathbb{R})$, but this is only possible if
$A=\lambda I_{n\times n}$ for some $\lambda\in\mathbb{R}$. The John ellipsoid
is thus of the type $B^{2n}(\lambda^{-1/2})$ for some $\lambda\geq1$ and this
concludes the proof since $\lambda>1$ is excluded (maximality of the John
ellipsoid forces $\lambda=1$) .
\end{proof}

\begin{theorem}
\label{propblob1}Let $X=\{x:Ax\cdot x\leq\hbar\}$ for some $A=A^{T}>0$. The
product $X\times X^{\hbar}$ contains a unique quantum blob, namely the John
ellipsoid
\begin{equation}
\operatorname{John}(X\times X^{\hbar})=M_{A^{-1/2}}(B^{2n}(\sqrt{\hbar}))
\label{qmab}%
\end{equation}
where $M_{A^{-1/2}}$ is the block-diagonal symplectic matrix%
\begin{equation}
M_{A^{1/2}}=%
\begin{pmatrix}
A^{-1/2} & 0\\
0 & A^{1/2}%
\end{pmatrix}
\label{ma}%
\end{equation}
and we have%
\begin{equation}
\Pi_{X^{\prime}}\Omega=X\ \ \ ,\ \ \Pi_{\ell P}\Omega=X^{\hbar}. \label{piq}%
\end{equation}

\end{theorem}

\begin{proof}
We have $X^{\hbar}=A^{1/2}(B_{P}^{n}(\sqrt{\hbar}))$ hence
\[
X\times X^{\hbar}=M_{A^{1/2}}(B_{X}^{n}(\sqrt{\hbar})\times B_{P}^{n}%
(\sqrt{\hbar}))
\]
and the claim follows from the fact that the John ellipsoid \ of $B_{X}%
^{n}(\sqrt{\hbar})\times B_{P}^{n}(\sqrt{\hbar})$ is $B^{2n}(\sqrt{\hbar})$
(see Lemma \ref{LemLoewner} above). The matrix $M_{A^{-1/2}}$ is symplectic,
so this ellipsoid is a quantum blob. However, every quantum blob has volume
$\operatorname{vol}B^{2n}(\sqrt{\hbar})$, which is precisely the volume of the
above John ellipsoid. Therefore any quantum blob contained in $X\times
X^{\hbar}$ would be a maximal-volume ellipsoid contained in $X\times X^{\hbar
}$. By uniqueness of the John ellipsoid, it must coincide with the one above.
That the orthogonal projections of $\operatorname*{QB}(S)$ on the $x$ and $p$
spaces are $X$ and $X^{\hbar}$ is obvious. The uniqueness of quantum blob
contained in $X\times X^{\hbar}$ follows from the uniqueness of the John
ellipsoid and the fact that all quantum blobs have the same volume.
\end{proof}

Notice that we have
\[
c(\operatorname{John}(X\times X^{\hbar}))=\pi\hbar
\]
for every symplectic capacity $c.$

\subsection{The general case}

Let us now consider the general case, that of an arbitrary polar quantum pair
$(X,P)$ with%
\begin{equation}
X^{\hbar}\subset\ P=\{p:Bp\cdot p\leq\hbar\}. \label{inclus}%
\end{equation}
By an argument similar to that above the John ellipsoid $\Omega
=\operatorname{John}(X\times P)$ is%
\begin{equation}
\Omega=K_{AB}\left(  (B^{2n}(\sqrt{\hbar})\right)  \text{ \ },K_{AB}=%
\begin{pmatrix}
A^{-1/2} & 0\\
0 & B^{-1/2}%
\end{pmatrix}
\label{omkab}%
\end{equation}
that \ is, in coordinates,%
\begin{equation}
\Omega=\left\{  z:Mz\cdot z\leq\hbar\right\}  \text{ \ },\text{ \ }M=%
\begin{pmatrix}
A & 0\\
0 & B
\end{pmatrix}
. \label{omam}%
\end{equation}
The orthogonal projections of $\Omega$ on the $x$ and $p$ spaces \ are
\begin{equation}
\Pi_{X}\Omega=X\text{ \ \textit{and}}\ \ \Pi_{P}\Omega=P. \label{projab}%
\end{equation}

Let us characterize the quantum blobs contained in $\operatorname{John}%
(X\times P)$. The first statement is just a consequence of Theorem
\ref{propblob1}. It says that we may rescale the John ellipsoid of $X\times
X^{\hbar}$ inside $X\times P$ to obtain new quantum blobs.

\begin{theorem}
\label{ThmJohnBlob}Let $(X,P)$ be a quantum polar pair. (i) Let $\lambda
_{\max}$ be defined by
\[
\lambda_{\max}=\sup\{\lambda\geq1:\lambda X^{\hbar}\subset P\}
\]
and set $M_{\lambda}=%
\begin{pmatrix}
\lambda^{-1}I & 0\\
0 & \lambda I
\end{pmatrix}
$. For $1\leq\lambda\leq\lambda_{\max}$ the ellipsoid%
\[
M_{\lambda}\operatorname{John}(X\times X^{\hbar})=M_{\lambda A^{-1/2}}%
(B^{2n}(\sqrt{\hbar}))
\]
is a quantum blob contained in $\operatorname{John}(X\times P)$. (ii) More
generally,%
\[
M_{A^{-1/2}B^{-1/2}}\operatorname{John}(X\times X^{\hbar})\subset
\operatorname{John}(X\times P)
\]
is a quantum blob
\end{theorem}

\begin{proof}
(i) Let $1\leq\lambda\leq\lambda_{\max}$. We have $\lambda^{-1}\leq1$ hence%
\begin{align*}
M_{\lambda}\operatorname{John}(X\times X^{\hbar})  &  =\operatorname{John}%
(\lambda^{-1}X\times\lambda X^{\hbar})\\
&  \subset\operatorname{John}(X\times P).
\end{align*}
The statement follows from $M_{\lambda}\in\operatorname*{Sp}(n)$ and formula
(\ref{qmab}) in Theorem \ref{propblob1}. (ii) It suffices to prove that
\[
M_{A^{-1/2}B^{-1/2}}(X\times X^{\hbar})\subset M_{A^{-1/2}B^{-1/2}}(X\times
P).
\]
We have
\[
M_{A^{-1/2}B^{-1/2}}=%
\begin{pmatrix}
B^{1/2}A^{1/2} & 0\\
0 & B^{-1/2}A^{-1/2}%
\end{pmatrix}
\in\operatorname*{Sp}(n)
\]
hence it suffices to show that $B^{1/2}A^{1/2}(X)\subset X$ and $B^{-1/2}%
A^{-1/2}(X^{\hbar})\subset P$. We have%
\[
B^{1/2}A^{1/2}(X)=B^{1/2}A^{1/2}(A^{-1/2}B^{n}(\sqrt{\hbar}))=B^{1/2}%
(B^{n}(\sqrt{\hbar})))
\]
hence $B^{1/2}A^{1/2}(X)=P^{\hbar}\subset X$ since $X^{\hbar}\subset P$. On
the other hand,
\[
B^{-1/2}A^{-1/2}(X^{\hbar})=B^{-1/2}A^{-1/2}(A^{1/2}B^{n}(\sqrt{\hbar
}))=B^{-1/2}(B^{n}(\sqrt{\hbar}))
\]
that is, $B^{-1/2}A^{-1/2}(X^{\hbar})=P$ so we are done.
\end{proof}

\section{From Indeterminacy to Uncertainty\label{sec5}}

We have so far discussed quantum indeterminacy from a geometric and
topological viewpoint. We now establish the connection between these
constructions and the uncertainty principle in its traditional form.

\emph{From now on we will view} $\operatorname{John}(X\times P)$ \emph{as the
canonical covariance-type ellipsoid associated with the data of }$X$\emph{ and
}$P$\emph{.}

We emphasize that, for a general quantum polar pair, the matrix $\Sigma$
defined by the John ellipsoid is a geometrically defined covariance-type
matrix and need not be the statistical covariance matrix of an underlying
quantum state. The latter interpretation requires, in particular, the
existence of second moments. The Gaussian case is distinguished: for the
Gaussian state naturally associated with a quantum blob (and, more generally,
with the corresponding Gaussian covariance ellipsoid), $\Sigma$ is precisely
the statistical covariance matrix. Thus the Gaussian setting provides an exact
bridge between the geometric construction developed here and the usual
statistical formulation of quantum uncertainty. On the other hand, the datum
of $X\times P$ does not in general give us any information on whether the
state under consideration is pure or mixed: pure quantum states can have
arbitrarily large covariance matrices. \ Here is a textbook example: Let
$h_{m}$ denote the $m$-th normalized Hermite function on $\mathbb{R}$, and
define the pure non-Gaussian state \ $\psi_{m}=$ $h_{m}\otimes\cdot\cdot
\cdot\otimes h_{m}$. The latter has covariance matrix%
\[
\Sigma_{m}=\hbar\left(  m+\frac{1}{2}\right)  I_{2n}%
\]
and corresponding covariance ellipsoid%
\[
\Omega_{m}=B^{2n}(R_{m})\text{ \ },\text{ \ }R_{m}^{2}=2\hbar\left(
m+\frac{1}{2}\right)
\]
becomes arbitrarily large as $m\rightarrow\infty$.

\begin{theorem}
Let $\Omega=\operatorname{John}(X\times P)$ be the John ellipsoid of a quantum
polar pair $(X\times P)$. Let $\Sigma$%
%TCIMACRO{\TEXTsymbol{>}}%
%BeginExpansion
$>$%
%EndExpansion
$0$ be the matrix determined by the condition%
\begin{equation}
\Omega=\{z:\tfrac{1}{2}\Sigma^{-1}z\cdot z\leq1\}.
\end{equation}
(i) Then we have
\begin{equation}
c(\Omega)\geq\pi\hbar\label{charlyne}%
\end{equation}
for every symplectic capacity $c$; (ii) the matrix $\Sigma$ satisfies the
quantum condition%
\begin{equation}
\Sigma+\frac{i\hbar}{2}J\geq0. \label{Quantum}%
\end{equation}
If $c(\Omega)=\pi\hbar$, then the quantum blob contained in $\Omega$ is unique.
\end{theorem}

\begin{proof}
(i) In view of Theorem \ref{ThmJohnBlob} $\Omega$ contains at least one
quantum blob, hence $c(\Omega)\geq\pi\hbar$. For a proof of the uniqueness of
the quantum blob in $\Omega$ when $c(\Omega)=\pi\hbar$ \ see
\cite{Bullsci,goluPR}. (ii) Let us show that (\ref{charlyne}) implies the
quantum condition (\ref{Quantum}). Setting $M=\frac{1}{2}\Sigma^{-1}$ the
ellipsoid $\Omega$ is the set of all $z\in\mathbb{R}^{2n}$ such that $Mz\cdot
z\leq1$ and the condition $\Sigma+\frac{i\hbar}{2}J\geq0$ is equivalent to
$\frac{1}{2}M^{-1}+\frac{i\hbar}{2}J\geq0$. Using a Williamson diagonalization
\cite{Birk,HZ} of $M$ this is equivalent to $\frac{1}{2}D^{-1}+\frac{i\hbar
}{2}J\geq0$ where
\[
D=%
\begin{pmatrix}
\Lambda^{\sigma} & 0\\
0 & \Lambda^{\sigma}%
\end{pmatrix}
\text{ \ , \ }\Lambda^{\sigma}=\operatorname*{diag}(\lambda_{1}^{\sigma
},...,\lambda_{n}^{\sigma})
\]
the $\lambda_{j}^{\sigma}>0$ being the symplectic eigenvalues of $M$
(\textit{i.e}. $\pm i\lambda_{j}^{\sigma}$ is an eigenvalue of $JM$.) It
follows that the characteristic polynomial of $\frac{1}{2}M^{-1}+\frac{i\hbar
}{2}J$ is the product $P(t)=P_{1}(t)\cdot\cdot\cdot P_{n}(t)$ where%
\[
P_{\alpha}(t)=t^{2}-(\lambda_{\alpha}^{\sigma})^{-1}t+\tfrac{1}{4}%
(\lambda_{\alpha}^{\sigma})^{-2}-\tfrac{1}{4}\hbar^{2}\text{.}%
\]
The eigenvalues of the matrix $\frac{1}{2}M^{-1}+\frac{i\hbar}{2}J$ are thus
the real numbers $\frac{1}{2}[(\lambda_{j}^{\sigma})^{-1}\pm\hbar]$ hence that
matrix is non-negative if and only if $\lambda_{j}^{\sigma}\leq\frac{1}%
{2}(\frac{1}{2}\hbar)^{-1}$ for every $j$, that is, if and only if
$\lambda_{\max}^{\sigma}\leq\frac{1}{2}(\frac{1}{2}\hbar)^{-1}$; this is in
turn equivalent to
\[
c(\Omega)=\pi/\lambda_{\max}^{\sigma}\geq\pi\hbar.
\]

\end{proof}

The condition (\ref{Quantum}) has been studied by many authors in the context
of the positivity of trace class operators: see
\cite{Dutta,Sudar,Birkbis,goluPR} and Narcowich and O'Connell
\cite{Narconnell}).

Let us now derive the Robertson--Schr\"{o}dinger inequalities from the
construction above. Write the symmetric matrix $\Sigma$ in block form
\begin{equation}
\Sigma=
\begin{pmatrix}
\Sigma_{XX} & \Sigma_{XP}\\
\Sigma_{PX} & \Sigma_{PP}%
\end{pmatrix}
,\qquad\Sigma_{PX}=\Sigma_{XP}^{T}.
\end{equation}
When $\Sigma$ is interpreted as a covariance matrix, its diagonal entries are
$(\Delta x_{j})^{2}$ and $(\Delta p_{j})^{2}$, while the corresponding
diagonal entries of $\Sigma_{XP}$ are the covariances $\Delta(x_{j},p_{j})$.

\begin{corollary}
The quantum condition \eqref{Quantum} implies the Robertson--Schr\"{o}dinger
inequalities
\begin{equation}
(\Delta x_{j})^{2}(\Delta p_{j})^{2}\geq\Delta(x_{j},p_{j})^{2}+\frac
{\hbar^{2}}{4},\qquad1\leq j\leq n. \label{RS2}%
\end{equation}

\end{corollary}

\begin{proof}
For each $j$, take the principal $2\times2$ Hermitian submatrix of
$\Sigma+\frac{i\hbar}{2}J$ corresponding to the conjugate variables
$(x_{j},p_{j})$:
\[%
\begin{pmatrix}
(\Delta x_{j})^{2} & \Delta(x_{j},p_{j})+\frac{i\hbar}{2}\\
\Delta(x_{j},p_{j})-\frac{i\hbar}{2} & (\Delta p_{j})^{2}%
\end{pmatrix}
.
\]
It is non-negative because \eqref{Quantum} holds. Its determinant is therefore
non-negative, which gives
\[
(\Delta x_{j})^{2}(\Delta p_{j})^{2}-\Delta(x_{j},p_{j})^{2}-\frac{\hbar^{2}%
}{4}\geq0.
\]
This is precisely (\ref{RS2}). Notice that the full matrix condition
(\ref{Quantum}) is stronger than the collection of these $n$ scalar
inequalities because it also controls correlations between different degrees
of freedom.
\end{proof}

\section{Concluding Remarks\label{secConclusion}}

The purpose of this paper has been to separate a geometric notion of quantum
indeterminacy from its usual statistical formulation. The basic input is the
polar duality inclusion $X^{\hbar}\subset P$ for centrally symmetric position
and momentum bodies. This condition is formulated directly in terms of
phase-space geometry and the quantum of action $\hbar$, without initially
invoking means, variances, or a probability distribution.

The link with the standard uncertainty principle is recovered rather than
postulated. Quantum blobs provide the symplectic motivation for the polar
condition; John ellipsoids then associate with a quantum polar pair a
canonical covariance-type ellipsoid. The existence of a quantum blob inside
this ellipsoid yields the capacity bound $c(\Omega)\geq\pi\hbar$, which is
equivalent to $\Sigma+\frac{i\hbar}{2}J\geq0$. When $\Sigma$ is interpreted as
a covariance matrix, the Robertson--Schr\"{o}dinger inequalities follow
immediately. In the Gaussian case this identification is exact: $\Sigma$ is
the covariance matrix of the corresponding Gaussian state, so that the matrix
condition is precisely the standard quantum covariance condition. Thus the
Gaussian case provides the direct bridge between the present geometric
indeterminacy condition and the conventional statistical uncertainty
relations. For general quantum polar pairs, by contrast, $\Sigma$ should be
regarded as a covariance-type geometric matrix; in this sense the present
condition is a geometric extension of the uncertainty structure rather than a
claim of logical equivalence with the statistical Robertson--Schr\"{o}dinger
inequalities for arbitrary states. The relation between $\Sigma$ and the
empirical covariance matrix will be clarified in forthcoming work; it requires
the use of statistical tools which lie outside the scope of the present work
If only for reasons of length).

Several limitations should be kept explicit. Central symmetry and centering
are part of the present framework, although translations and momentum boosts
can be removed by recentering the data. Moreover, the canonical John ellipsoid
is a covariance-type geometric object and need not coincide with the
statistical covariance ellipsoid of an arbitrary quantum state. The Gaussian
case is distinguished because the two descriptions coincide: the matrix
$\Sigma$ defining the ellipsoid is then the covariance matrix of the
corresponding Gaussian state. Extending the construction to
non-centrally-symmetric data and clarifying its operational interpretation for
general experimental clouds are natural directions for further work.

\section*{APPENDIX\ A: Gromov's Theorem and Symplectic Capacities}

For details and proofs we refer to the treatise \cite{HZ,Polter} and to our
\textit{Phys. Reps.} paper \cite{goluPR}. We denote by $\operatorname*{Symp}%
(n)$ the group of all symplectomorphisms of $(\mathbb{R}_{z}^{2n},\sigma)$:
$f\in\operatorname*{Symp}(n)$ if and only if $f$ is a diffeomorphism of
$\mathbb{R}_{z}^{2n}$ and $f^{\ast}\sigma=\sigma$; equivalently $f$ is
bijective, infinitely differentiable and with infinitely differentiable
inverse, and the Jacobian matrix $Df(z)\in\operatorname*{Sp}(n)$ for every $z$.

Let us denote by $Z_{j}^{2n}(R)$ the phase space cylinder in defined by
\[
Z_{j}^{2n}(R)=\left\{  z=(x,p)\in\mathbb{R}^{2n}:x_{j}^{2}+p_{j}^{2}\leq
R^{2}\right\}  ;
\]
thus, the plane of conjugate coordinates $x_{j},p_{j}$ cuts $Z_{j}^{2n}(R)$
orthogonally, along the disk $x_{j}^{2}+p_{j}^{2}\leq R^{2}$ . Also consider
the phase space
\[
B^{2n}(R)=\left\{  (x,p)\in\mathbb{R}^{2n}:x^{2}+p^{2}\leq R^{2}\right\}  ;
\]
Mischa Gromov \cite{gr85} proved, using methods from complex analysis and
symplectic topology, the following deep result:

\begin{theorem}
[Gromov]If there exists $f\in$ $\operatorname*{Symp}(n)$ such that
$f(B^{2n}(r))\subset Z_{j}^{2n}(R)$ for some $j=1,...,n$, then we must have
$r\leq R$.
\end{theorem}

A consequence of Gromov's theorem (often viewed as an alternative formulation
thereof) is the following: for every phase space symplectomorphism
$f\in\operatorname*{Symp}(n)$ the area of the orthogonal projection of
$f(B^{2n}(R))$ on any of the $x_{j},p_{j}$ planes is $\geq\pi R^{2}$.

Gromov's non-squeezing theorem guarantees the existence of symplectic
capacities. Consider now the two functions associating to any subset of
$\mathbb{R}^{2n}$ the non-negative numbers, or%
%TCIMACRO{\U{b4}}%
%BeginExpansion
\'{}%
%EndExpansion
, $+\infty$:
\begin{align}
c_{\min}(\Omega)  &  =\sup_{f\in\operatorname*{Symp}(n)}\left\{  \pi
R^{2}:f(B^{2n}(R))\subset\Omega\right\} \label{cmin}\\
c_{\max}(\Omega)  &  =\inf_{f\in\operatorname*{Symp}(n)}\left\{  \pi
R^{2}:f(\Omega)\subset Z_{j}^{2n}(R)\right\}  . \label{cmax}%
\end{align}

The functions $c_{\min}$ and $c_{\max}$ have \ \ the following properties:

\begin{proposition}
Let $c=c_{\min}$ or $c_{\max}$. (i) $c$ is monotone:
\[
\Omega\subset\Omega^{\prime}\Longrightarrow c(\Omega)\leq c(\Omega^{\prime});
\]
(ii) $c$ is invariant by symplectomorphisms:
\[
c(f\Omega)=c(\Omega)\text{ for every }f\in\operatorname*{Symp}(n);
\]
(iii) $c$ is quadratically homogeneous: $c(\lambda\Omega)=\lambda^{2}%
c(\Omega)$ for every $\lambda\in\mathbb{R}$; (iv) Non-triviality:
\[
c(B^{2n}(R))=\pi R^{2}=c(Z_{j}^{2n}(R)).
\]

\end{proposition}

\begin{proof}
Properties (i) and (ii) are obvious. Let $c=c_{\min}$; we have
\begin{align*}
c_{\min}(\lambda\Omega)  &  =\sup_{f\in\operatorname*{Symp}(n)}\left\{  \pi
R^{2}:f(B^{2n}(R))\subset\lambda\Omega\right\} \\
&  =\sup_{f\in\operatorname*{Symp}(n)}\left\{  \pi R^{2}:(\lambda^{-1}%
f\lambda)\lambda^{-1}(B^{2n}(R))\subset\Omega\right\} \\
&  =\sup_{f\in\operatorname*{Symp}(n)}\left\{  \pi R^{2}:\lambda^{-1}%
f\lambda\lambda(B^{2n}(\lambda^{-2}R))\subset\Omega\right\}
\end{align*}
hence $c_{\min}(\lambda\Omega)=\lambda^{2}c_{\min}(\Omega)$ since
$\lambda^{-1}f\lambda\lambda f\in$ $\operatorname*{Symp}(n)$ if $f\in$
$\operatorname*{Symp}(n)$. The case $c=c_{\max}$ is treated similarly.
Property (iv) follows from Gromov's theorem.
\end{proof}

The case of ellipsoids is of particular interest, as all symplectic capacities
agree on them Let
\[
\Omega=\{z\in\mathbb{R}^{2n}:Mz\cdot z\leq\hbar
\]
where $M\in\operatorname*{Sym}(n,\mathbb{R})$, $M>0$. For every symplectic
capacity $c$ we have
\begin{equation}
c(\Omega)=\frac{\pi\hbar}{\lambda_{\max}^{\sigma}} \label{capellipsoid}%
\end{equation}
where $\lambda_{\max}^{\sigma}$ is the largest symplectic eigenvalue of the
matrix $M$ (the symplectic eigenvalues \cite{Birk} of $M$ are the numbers
$\lambda_{j}^{\sigma}>0$ ($1\leq j\leq n$) such that the $\pm i\lambda
_{j}^{\sigma}$ are the eigenvalues of $JM$).

Let $K\subset{\mathbb{R}}^{2n}$ be a compact convex body with smooth boundary
$\Sigma=\partial K$. That boundary carries at least one closed characteristic,
i.e. a periodic orbit of the Hamiltonian flow (see e.g. \cite{AM}.: since $K$
is convex and contains the origin, define the Minkowski functional
\[
j_{K}(z)=\operatorname{inf}\{r>0:z\in rK\}
\]
and set $H(z)=j_{K}(z)^{2}$. Then $H^{-1}(1)=\Sigma$ and the Hamiltonian
vector field is $X_{H}=J\nabla H.$

In \cite{HZ} Hofer--Zehnder construct a symplectic capacity $c_{^{\text{HZ}}}$
which measures sets in a dynamical way. It is defined as follows. Let $\Omega$
be an open set in $\mathbb{R}^{2n}$ and consider the class $\mathcal{H}%
(\Omega)$ of all Hamiltonian functions $H\geq0$ having the following three properties:

\begin{itemize}
\item $H$ vanishes outside $\Omega$ (and is hence bounded);

\item The critical values of $H$ are $0$ and $\max H$;

\item The flow $(f_{t}^{H})$ has no non-constant periodic orbit with period
$T\leq1$.
\end{itemize}

Then, by definition,%
\begin{equation}
c_{^{\text{HZ}}}(\Omega)=\sup\{\max H:H\in\mathcal{H}(\Omega)\}.
\label{defchz}%
\end{equation}

The Hofer--Zehnder capacity has the remarkable property that whenever $\Omega$
is a compact convex set in phase space then%
\begin{equation}
c_{^{\text{HZ}}}(\Omega)=\oint\nolimits_{\gamma_{\min}}pdx \label{chz}%
\end{equation}
where $pdx=p_{1}dx_{1}+\cdot\cdot\cdot+p_{n}dx_{n}$ and $\gamma_{\min}$ is the
shortest (positively oriented) Hamiltonian periodic orbit carried by the
boundary $\partial\Omega$ of $\Omega$. (In formula (\ref{chz}) the condition
that $\Omega$ be compact and convex is essential, see Hofer and Zehnder's very
illustrative \textquotedblleft Bordeaux bottle\textquotedblright\ example in
\cite{HZ}, p. 99).

\begin{acknowledgement}
The author has been financed by a "QUANTUM\ AUSTRIA" grant of the Austrian
Research Foundation FWF (Grant number PAT 2056623).
\end{acknowledgement}

\textbf{Data Availability. }No data has been used or created.

\textbf{Conflicts of interest.} There are no conflicts of interest.

\end{document}